\documentclass[12pt, draftclsnofoot, onecolumn]{IEEEtran}
\usepackage{bm,cite}
\usepackage{amsmath}
\usepackage{amsthm}
\usepackage{amsbsy}
\usepackage{epsfig}
\usepackage{latexsym}
\usepackage{psfrag}
\usepackage{amssymb} 
\usepackage{placeins}

\newtheorem{theorem}{Theorem}

\newtheorem{lemma}{Lemma}

\DeclareMathAlphabet{\mathbit}{OML}{cmr}{bx}{it}
\DeclareMathAlphabet{\mathsf}{OT1}{cmss}{m}{n}
\DeclareMathAlphabet{\mathTXf}{OT1}{cmss}{bx}{it}

\DeclareMathOperator{\diag}{diag}

\DeclareMathOperator{\ZF}{ZF}
\DeclareMathOperator{\DoF}{DoF}

\DeclareMathOperator{\Id}{\mathbf{I}}
\DeclareMathOperator{\CN}{\mathcal{N}_{\mathbb{C}}}
\DeclareMathOperator{\KMAT}{K_{\alpha}-\mathrm{MAT}}

\DeclareMathOperator{\Out}{\mathrm{Out}}
\DeclareMathOperator{\MAT}{\mathrm{MAT}}
 
\DeclareMathOperator{\ALTMAT}{\mathrm{A-MAT}}
\DeclareMathOperator{\AlphaMAT}{\alpha\mathrm{-MAT}}

\newcommand{\I}{\mathbf{I}} 
\newcommand{\norm}[1]{\lVert{#1}\rVert}

\newcommand{\Fro}{{\text{F}}}
\newcommand{\tr}{{\text{tr}}}

\newcommand{\E}{{\mathrm{E}}}

\newcommand{\He}{{{\mathrm{H}}}}

\newcommand{\xv}{\mathbf{x}}

\begin{document} 
\title{On the Degrees of Freedom of the $K$-User Time Correlated Broadcast Channel with Delayed CSIT}
\author{Paul de Kerret, Xinping Yi and David Gesbert\\Mobile Communications Department, Eurecom\\
Campus SophiaTech, 450 Route des Chappes, 06410 Biot, France\\\{dekerret,yix,gesbert\}@eurecom.fr}

\maketitle
\begin{abstract} 
The Degrees of Freedom (DoF) of a $K$-User MISO Broadcast Channel (BC) is studied when the Transmitter (TX) has access to a delayed channel estimate in addition to an imperfect estimate of the current channel. The current estimate could be for example obtained from prediction applied on past estimates, in the case where feedback delay is within the coherence time. Building on previous recent works on this setting with two users, the estimation error of the current channel is characterized by its scaling as~$P^{-\alpha}$ where~$\alpha=1$ (resp. $\alpha=0$) corresponds to an estimate being essentially perfect (resp. useless) in terms of DoF. In this work, we contribute to the characterization of the DoF region in such a setting by deriving an outerbound for the DoF region and by providing an achievable DoF region. The achievable DoF is obtained by developing a new alignment scheme, called the $\KMAT$ scheme, which builds upon both the principle of the MAT alignment scheme from Maddah-Ali and Tse and Zero-Forcing to achieve a larger DoF when the delayed CSIT received is correlated with the instantaneous channel state.  
\end{abstract} 
\IEEEpeerreviewmaketitle 
\section{Introduction}\label{se:SM}
The use of multiple-antenna has been recognized during the last decade as a key element to improve performance in wireless networks due to the possibility to achieve a larger number of Degrees-of-Freedom (DoF), or pre-log factor, by transmitting several independent data streams at the same time\cite{Telatar1999}. While in point-to-point MIMO systems, the maximal DoF can be achieved without Channel State Information (CSI) at the Transmitter (TX), the exploitation of the multiple-antennas at the TX to achieve a DoF larger than one in multiuser settings heavily relies on the availability of accurate-enough CSI at the TX (CSIT). For instance, it is well known that in the $K$-user Multiple-Input Single Output (MISO) Broadcast Channel (BC), the DoF is reduced from~$K$ to~$1$ in the absence of CSIT\cite{Jafar2005} while full DoF is preserved if the variance of the channel estimation error falls as $P^{-1}$ or faster, where~$P$ is the Signal-to-Noise Ratio (SNR)\cite{Jindal2006,Caire2010}. Similar conclusions have been obtained in more general settings \cite{Huang2012,Vaze2012b}.

Yet, the obtaining of an accurate-enough CSIT represents a challenge in many settings. Indeed, the channel estimate has to be fed back from the RXs which inevitably introduces some delays and some degradations. Therefore, a large literature has focused on the problem of designing efficient feedback schemes and evaluating the impact of imperfect CSIT [See \cite{Love2008,Jindal2006} and reference therein].

Recently, a new line of work was opened by the work from Maddah-Ali and Tse \cite{MaddahAli2010,MaddahAli2012}. Studying a $K$-user MISO BC, they showed that even completely outdated CSIT, in the sense that the feedback delay exceeds the coherence period of the channel, could still be used to achieve a larger DoF than in the absence of CSIT. This is accomplished through a space-time alignment of the interference referred in the literature as the \emph{MAT} alignment. Furthermore, if the channel matrices are independent and identically distributed over time and across the Receivers (RXs), the $\MAT$ scheme is then optimal in terms of DoF.

This new method of exploiting stale CSIT has attracted a large interest and has been extended to further network scenarios. In \cite{Vaze2011a,Abdoli2011a}, the approach is adapted to two-user and three-user settings with multiple-antenna at the RXs, and to Interference Channels (ICs) and X-channels in \cite{Vaze2012a,Maleki2012,Abdoli2011b,Tandon2012a}, among others. In \cite{Mohanty2012}, the IC with TXs having unequal CSIT is also investigated.

Going beyond completely outdated CSIT, settings with CSIT of alternating qualities have been investigated. In \cite{Lee2012}, a setting is studied in a block fading model where the CSIT is only accurate for some time slots and completely outdated during others. It is then shown that under some conditions the maximal DoF can still be achieved. Considering a more general CSIT model, the two-user MISO BC is studied in \cite{Tandon2012b} in the case where the CSIT relative to one user is alternatively perfect, completely outdated, or non-existent. It is then shown that the alternating between different CSIT configurations can lead to synergistic benefits.

Yet, a major restriction of these works is that they all consider the delayed CSIT as being completely uncorrelated with the instantaneous channel state. This assumption is lifted in \cite{Kobayashi2012} where an improved DoF is shown to be achievable in the case where the delayed CSIT is assumed to be possibly correlated with the current channel state. As a consequence, an imperfect estimate of the current channel can be obtained by prediction based on the delayed CSIT. Specifically, it is assumed that the channel estimation error resulting from the prediction based on the delayed CSIT scales as $P^{-\alpha}$ with $\alpha\geq 0$ being the \emph{CSIT quality exponent}. Thus, when~$\alpha$ is equal to one, the imperfect estimate of the current channel is essentially perfect in terms of DoF. On the opposite when~$\alpha$ tends to zero, the estimate of the current channel is essentially useless. 

Building on the approach developed in \cite{Kobayashi2012}, the scheme was improved to reach the maximal DoF in a two-user MISO scenario\cite{Yang2012,Gou2012}. The scheme achieving the optimal DoF region in the two-user MISO BC is referred hereafter as the $\AlphaMAT$ scheme. This approach has then been extended to imperfect delayed CSIT in \cite{Chen2012a,Chen2012b} and to two-user MIMO BC and IC in \cite{Yi2012}. The study of delayed CSIT correlated to the instantaneous channel state has always remained restricted to the two-user case and the results do not trivially extend to more users. Finding the DoF region and extending the $\AlphaMAT$ alignment to more users is precisely the goal of this work.

Specifically, our main contributions are as follows.
\begin{itemize}
\item As a preliminary step, we develop a new alignment scheme, called the $\ALTMAT$ scheme, to exploit completely outdated CSIT. This scheme can be seen as an extension of the alternative version of $\MAT$ for the two-user case and is more adapted to the combined use of ZF and alignment based on delayed CSIT. Yet, its performances are suboptimal.
\item We derive an outerbound for the $K$-user MISO broadcast channel with delayed CSIT and imperfect current CSIT with quality exponent~$\alpha$.
\item We develop a new scheme which combines the $\ALTMAT$ alignment scheme and Zero-Forcing (ZF) in such a way that the sum DoF takes the simple form $(1-\alpha)\DoF^{\ALTMAT}+\alpha \DoF^{\ZF}$, where $\DoF^{\ALTMAT}$ and $\DoF^{\ZF}$ are the sum DoF achieved respectively with the $\ALTMAT$ scheme and with ZF.
\end{itemize}

\emph{Notations:} The complex circularly invariant Gaussian distribution of mean~$\mu$ and variance~$\sigma^2$ is denoted by~$\CN(0,\sigma^2)$. $f(x)\sim g(x)$ denotes the fact that~$\lim_{x\rightarrow \infty}\frac{f(x)}{g(x)}=C$ with $C\neq 0$. The $j$th element of the $i$th row of the matrix~$\mathbf{A}$ is denoted by $\{\mathbf{A}\}_{ij}$. The function $\log$ represents the logarithm with base~$2$ and~$\norm{\mathbf{A}}_{\Fro}$ the Frobenius norm of the matrix~$\mathbf{A}$. $\mathbf{A}\succeq\mathbf{0}$ is used to represent the fact that the matrix $\mathbf{A}$ is positive semidefinite while $\mathbf{A}\succeq \mathbf{B}$ denotes the that $\mathbf{A}- \mathbf{B}\succeq \mathbf{0}$. If~$\mathbf{A}$ is a positive definite matrix, $\mathbf{A}^{1/2}$ denotes the unique lower triangular matrix with strictly positive coefficient obtained via the Cholesky factorization such that~$\mathbf{A}=\mathbf{A}^{1/2}(\mathbf{A}^{1/2})^{\He}$. We write wlog for \emph{without loss of generality} and i.i.d. for \emph{independently and identically distributed}.
\section{System Model}\label{se:SM}

\subsection{$K$-User MISO Broadcast Channel}

This work considers a $K$-User MISO BC where the TX is equipped with $M$~antennas and serves $K$~single-antenna users. We assume furthermore that $M\geq K$. At any time~$t$, the signal received at RX~$i$ can be written as
\begin{equation}
y_i(t)=\bm{h}_i^{\He}(t)\bm{x}(t)+z_i(t)
\label{eq:SM_1}
\end{equation}
where $\bm{h}_i^{\He}\in \mathbb{C}^{1\times M}$ is the channel to user~$i$ at time~$t$,~$\bm{x}\in \mathbb{C}^{M\times 1}$ is the transmitted signal, and $z_i(t)\in \mathbb{C}$ is the additive noise at RX~$i$, independent of the channel and the transmitted signal and distributed as~$\CN(0,1)$. Furthermore, the transmitted signal~$\bm{x}(t)$ fulfills the average power constraint~$\E[\|\bm{x}(t)\|^2]\leq P$. 

We define further the channel matrix~$\mathbf{H}\triangleq [\bm{h}_1,\ldots,\bm{h}_K]^{\He}\in \mathbb{C}^{K\times M}$ and introduce the notation~$\mathcal{H}^{t}\triangleq \{\mathbf{H}(k)\}_{k=1}^{k=t}$. 
The channel is assumed to be drawn from a continuous ergodic distribution such that all the channel matrices and all their submatrices are full rank.

\subsection{Delayed CSIT with Correlation in Time}
The considered CSIT model builds on the delayed CSIT model introduced in \cite{MaddahAli2010} and generalized to account for time correlation in \cite{Kobayashi2012}. According to this model, the TX has access at time~$t$ to the delayed CSI. It takes the form of the CSI up to time $t-1$ which is denoted by~$\mathcal{H}^{t-1}$. Furthermore, exploiting the correlation in time between the delayed CSI~$\mathcal{H}^{t-1}$ and the current channel state~$\mathbf{H}(t)$, the TX produces an imperfect estimate of the channel state denoted by~$\hat{\mathbf{H}}(t)$. This channel estimate is then modeled such that
\begin{equation}
\mathbf{H}(t)=\hat{\mathbf{H}}(t)+\tilde{\mathbf{H}}(t)
\label{eq:SM_2}
\end{equation}
where the channel estimate and the channel estimation error are independent, the channel estimation error~$\tilde{\mathbf{H}}(t)$ has its elements i.i.d. $\CN(0,\sigma^2)$ while the elements of the channel estimate~$\hat{\mathbf{H}}(t)$ are assumed to have a variance equal to~$1-\sigma^2$. We further define~$\hat{\mathcal{H}}^t\triangleq \{\hat{\mathbf{H}}(k)\}_{k=1}^{k=t}$ and $\tilde{\mathcal{H}}^t\triangleq \{\tilde{\mathbf{H}}(k)\}_{k=1}^{k=t}$.

It is also assumed that the channel state~$\mathbf{H}(t)$ is independent of the pair~$(\hat{\mathcal{H}}^{t-1},\tilde{\mathcal{H}}^{t-1})$ when conditioned on~$\hat{\mathbf{H}}(t)$.

The variance~$\sigma^2$ of the estimation error is parameterized as a function of the SNR~$P$ such that~$\sigma^2= P^{-\alpha}$ where we have defined the \emph{CSIT quality exponent}~$\alpha$ as
\begin{equation}
\alpha\triangleq\lim_{P\rightarrow \infty} \frac{-\log(\sigma^2)}{\log(P)}.
\label{eq:SM_3}
\end{equation}
Note that from a DoF perspective, we can restrict ourselves to $\alpha \in [0,1]$ since an estimation/quantization error scaling as~$P^{-1}$ is essentially perfect while an estimation error scaling as~$P^{0}$ is essentially useless in terms of DoF.

\emph{Remark:} This suggests that in order to keep the rate scaling in the SNR, and under a given time-correlation model, the feedback delay as a fraction of the correlation time must shrink as the SNR increases (e.g., the terminal velocity must decrease).

Note furthermore that for any ZF precoded vector~$\bm{u}$ such that $\hat{\bm{h}}_i^{\He}\bm{u}=0$, it can easily be shown that~$\E[|\bm{h}^{\He}_i\bm{u}|^2]\sim P^{-\alpha}\E[\|\bm{u}\|^2]$.

Following the conventional assumption from the literature of delayed CSIT (e.g., in \cite{MaddahAli2012}), all the RXs are assumed to receive with a certain delay both the perfect multiuser CSI and the imperfect CSI. This CSI is used only for the RX to decode its data symbols such that the only limitation for this delay lies in the delay requirement of the data transmitted. The CSI at the RX side could for example be obtained if each user broadcasts is CSI implying that the other RXs can obtain the same CSI as the TX. Another solution is to simply let the TX send its perfect delayed CSIT to all the RXs\cite{Xu2012}.

\subsection{Degrees-of-Freedom Analysis}
Albeit an incomplete measure of system performance, the DoF offers the unique advantage of allowing for analytical tractability for even complex network models and feedback scenarios such as this one. Let us denote by~$\mathcal{D}^{*}$ the DoF-region, which is defined as follows.
\begin{equation}
\mathcal{D}^{*}\triangleq\left\{(d_1,d_2,\ldots,d_K)|\exists (R_1(P),\ldots,R_K(P)\in \mathcal{C}(P)\text{ , s.t. }\forall i=1,\!\ldots\!,K,d_i\!=\!\lim_{P\rightarrow \infty}\frac{R_i(P)}{\log(P)}\right\}
\label{eq:SM_4}
\end{equation}
where $\mathcal{C}(P)$ is the capacity region. Furthermore, the maximal sum DoF will also be of particular interest in this work. We denote it by~$\DoF^{*}$ and define it such that
\begin{equation}
\DoF^{*}\triangleq \max_{(d_1,\ldots, d_K)\in \mathcal{D}^{*}}\sum_{i=1}^K d_i.
\label{eq:SM_5}
\end{equation} 
\section{Main Results}\label{se:main}
We provide in this section our main results.
\subsection{Outerbound}\label{se:main:out}

We start by describing an outerbound for the DoF region, which will then be proven in Section~\ref{se:Outerbound}.
\begin{theorem}
In the $K$-user MISO BC with perfect delayed CSIT and current CSIT with quality exponent~$\alpha$, the DoF region $\mathcal{D}^{*}$ is outerbounded by~$\mathcal{D}^{\Out}$ defined by
\begin{align}
\forall \pi \in \mathcal{S}_p, p\in \{2,\ldots,K\}, ~~~~ \sum_{k=1}^p \frac{d_{\pi(k)}}{k} &\leq 1+\alpha\sum_{k=2}^{p} \frac{1}{k}\\
\forall i\in \{1,\ldots,K\}, ~~~0\leq d_i&\leq 1.
\end{align}
where $\mathcal{S}_p$ is the symmetric group containing all the permutations of $\{1,\ldots,p\}$.
In turn, the sum DoF is upperbounded by~$\DoF^{\Out}$ defined as
\begin{align}
\DoF^{\Out}=\frac{K\left(1+\alpha\sum_{k=2}^{K} \frac{1}{k}\right)}{\sum_{k=1}^{K} \frac{1}{k}}.
\label{eq:Main_1}
\end{align}
\label{thm_out}
\end{theorem}
\begin{proof}
The detailed proof is provided in Section~\ref{se:Outerbound}.
\end{proof}
It can be seen that this bound subsumes several known outerbounds from the literature. For~$\alpha=0$, it coincides with the optimal DoF achieved by the MAT algorithm while for~$\alpha=1$, the DoF in a MISO BC with perfect CSIT is obtained. Finally, for~$K=2$, this outerbound simplifies to the optimal DoF region provided in~\cite{Yang2012}.

\subsection{Achievable DoF}\label{se:main:out}

The problem of constructing a scheme achieving the outerbound in Theorem~\ref{thm_out} is very intricate and remains open. This is due to the difficulty to combine ZF (which is optimal for $\alpha=1$) with the MAT scheme (optimal for $\alpha=0$). The scheme for the two-user case developed in \cite{Yang2012,Gou2012} avoids this problem by using an alternative version of the MAT scheme developped by Maddah-Ali and Tse in \cite{MaddahAli2010}. In contrast with the original MAT scheme, this alternative version can be nicely combined with ZF such that the optimal DoF could then be achieved\cite{Yang2012,Gou2012}. This alternative version does not seem applicable for more than two users. As a consequence, our first step has been to find a new alignment scheme based on completely outdated CSIT, which, to some extent, generalizes the alternative MAT version to the case of more users. This scheme, denoted hereafter as the $\ALTMAT$ scheme, is described in Section~\ref{se:ALTMAT} and shown to achieve the following DoF.

\begin{theorem}
In the $K$-user MISO BC with completely outdated CSIT ($\alpha=0$), the $\ALTMAT$ scheme achieves a sum DoF equal to 
\begin{equation}
\DoF^{\ALTMAT}=\frac{2K}{K+1}+\frac{1}{n}\left(\frac{2K-3+\frac{2}{K+1}}{\left(\frac{(K-1)}{2}\right)+\frac{1}{n}\left(\frac{K-1}{2}+1\right)+\frac{1}{n}\left(\frac{K(K+1)}{2}-K\right)}\right)
\label{eq:Main_2}
\end{equation}
where the number $n_{\mathrm{TS}}$ of time slots over which the $\ALTMAT$ scheme is spread is
\begin{equation}
n_{\mathrm{TS}}=n\frac{K(K-1)}{2}+\frac{K(K-1)}{2}+\frac{K^2(K+1)}{2}-K(K-1).
\end{equation}
Hence, it holds
\begin{equation}
\lim_{n_{\mathrm{TS}}\rightarrow \infty}\DoF^{\ALTMAT}=\frac{2K}{K+1}.
\end{equation}
\label{thm_ALTMAT}
\end{theorem} 
The $\ALTMAT$ scheme can easily be adapted to exploit the correlation between the delayed CSIT and the instantaneous channel state. The modified scheme, denoted as the $\KMAT$ scheme, will then be shown in Section~\ref{se:KMAT} to achieve the following DoF.
\begin{theorem}
In the $K$-user MISO BC with perfect delayed CSIT and current CSIT with quality exponent~$\alpha$, the DoF achieved with the $\KMAT$ scheme is equal to
\begin{equation}
\DoF^{\KMAT}=(1-\alpha)\DoF^{\ALTMAT}+\alpha\DoF^{\ZF}
\label{eq:Main_3}
\end{equation}
with $\DoF^{\ZF}=K$.
\label{thm_KMAT}
\end{theorem}
The DoF achieved with ZF for the CSIT quality exponent $\alpha$ is well known to be equal to the second term of~\eqref{eq:Main_3} \cite{Jindal2006}. Hence, the $\KMAT$ scheme outperforms ZF and appears as a robust ZF scheme with respect to delay in the CSIT. The first term of \eqref{eq:Main_3} is the DoF improvement.

\section{The $\ALTMAT$ Scheme}\label{se:ALTMAT}


Similarly to the MAT scheme, the $\ALTMAT$ scheme does not exploit the correlation in time and hence treats the estimate as completely ``stale". Although suboptimal, the $\ALTMAT$ scheme can be easily adapted to exploit the time-correlation and henceforth will be a key component to develop a scheme which outperforms both $\MAT$ and ZF when $\alpha>0$. Similarly to \cite{MaddahAli2012}, a DoF strictly larger than one will be achieved by exploiting the broadcast nature of the channel. This means that a message destined to $j$~users (called order-$j$ messages) will be overheard by another $K-j$~users, hence providing side information which can be exploited. As a consequence, we will also define $\DoF_j$ as the DoF with which order-$j$ messages are transmitted. Note that with this notation, our objective is to transmit order-$1$ messages and to maximize $\DoF_1$.

When no confusion is possible, we omit to mention the dependency of the channels as a function of the time $t$.

\subsection{Example of the $\ALTMAT$ Scheme for $K=3$}\label{se:ALTMAT:3users}

The $\ALTMAT$ scheme consists of one initialization step, followed by a number of ``main iteration" steps and is ended by a termination step.
\begin{itemize}
\item \emph{Step~$1$--Initialization--} This step consists of $3$ time slots and takes as input $4$ order-$1$ symbols for every user. During the first time slot, the vector $\bm{u}_1\in \mathbb{C}^{2\times 1}$ containing $2$~data symbols for RX~$1$ and the vector~$\bm{u}_2\in \mathbb{C}^{2\times 1}$ containing $2$ data symbols for RX~$2$ are transmitted. The received signal at RX~$i$ can then be written as
\begin{equation}
y_i=\bm{h}_i^{\He}\bm{u}_1+\bm{h}_i^{\He}\bm{u}_{2}+z_i. 
\label{eq:ALTMAT_3}
\end{equation}
Following the same philosophy as the alternative form of the MAT scheme \cite{MaddahAli2012}, the interferences $\bm{h}_1^{\He}\bm{u}_{2}$ and $\bm{h}_2^{\He}\bm{u}_{1}$ are transmitted to both RX~$1$ and RX~$2$. Indeed, these equations are needed at both RXs because they represent, for one of them, the received interference, and for the other, a second independent observation of the desired signal. Hence, the transmission of the $4$ order-$1$ data symbols has been replaced by the transmission of $2$~order-$2$ data symbols. During the second (resp. the third) time slot, the same transmission scheme is used to transmit to RX~$2$ and RX~$3$ (resp. RX~$3$ and RX~$1$).
\item \emph{From step~$2$ to step~$n+1$--Main iteration step--} We assume that $6$ order-$2$ data symbols need to be transmitted to every user from the previous step. This phase is spread over $6$ time slots and takes as input $3$ order-$1$ messages for each user as well as the $6$ order-$2$ messages from the previous step.

In the first time slot, $3$ order-$1$ messages are transmitted to RX~$1$ while $2$~order-$2$ messages are transmitted to RX~$2$ and RX~$3$. We define the vector~$\bm{u}_1\in \mathbb{C}^{3\times 1}$ containing the $3$ order-$1$ messages and the vector~$\bm{u}_{23}\in \mathbb{C}^{2\times 1}$ containing the two order-$2$ messages. The received signal at RX~$i$ reads then as
\begin{equation}
y_i=\bm{h}_i^{\He}\bm{u}_1+\bm{h}_i^{\He}\bm{u}_{23}+z_i. 
\label{eq:ALTMAT_4}
\end{equation}
Let the interference $\bm{h}_1^{\He}\bm{u}_{23}$ be transmitted to all the RXs, the interference $\bm{h}_2^{\He}\bm{u}_{1}$ be transmitted to RX~$1$ and RX~$2$ and the interference $\bm{h}_3^{\He}\bm{u}_{1}$ to RX~$1$ and RX~$3$. Each RX can then decode its desired data symbols. Indeed, each RX could then remove the interference received as well as receive the right number of additional independent equations to decode its desired messages. Thus, $\bm{h}_1^{\He}\bm{u}_{23}$ can be seen as an order-$3$ message while $\bm{h}_2^{\He}\bm{u}_{1}$ and $\bm{h}_3^{\He}\bm{u}_{1}$ are order-$2$ messages. The transmission of the input data symbols has been replaced by the transmission of two order-$2$ messages and one order-$3$ message. During the two following time slots, the same transmission occurs after having permuted circularly the role of the RXs.

Finally, the three order-$3$ data symbols are broadcasted, which requires $3$ time slots. In total, $6$ order-$2$ data symbols have been transmitted and $9$ order-$1$ data symbols. At the same time, $6$ order-$2$ messages have been generated (from the overheard interference) and have to be transmitted in the following step.
\item \emph{Step~$n+2$-Termination-} At the beginning of this phase, $6$ order-$2$ data symbols have to be transmitted. This is carried out by simple broadcasting, and hence requires~$6$ time slots.
\end{itemize}
In total, $12+9n$ order-$1$ data symbols have been transmitted in $6+6n+6$ time slots. After simplifications, the DoF given in Theorem~\ref{thm_ALTMAT} is then obtained. As the number of main iteration steps~$n$ increases, the DoF converges to $3/2$.

The mains steps of the $\ALTMAT$ scheme for $K=3$ are illustrated in Fig.~\ref{ALTMAT_Scheme}. A particularity of $\ALTMAT$ is that symbols of different orders are sent at the same time. 

Note that the number of order-$2$ symbols transmitted is exactly equal to the number of order-$2$ messages created. This represents a particular case and for $K>3$, it will be necessary to consider several transmissions of symbols of different orders so as to reach an equilibrium where the number of data symbols of order-$j$ with $j\geq 2$ taken as input equals the number of symbols of order~$j$.

\begin{figure}
\centering
\includegraphics[width=1\columnwidth]{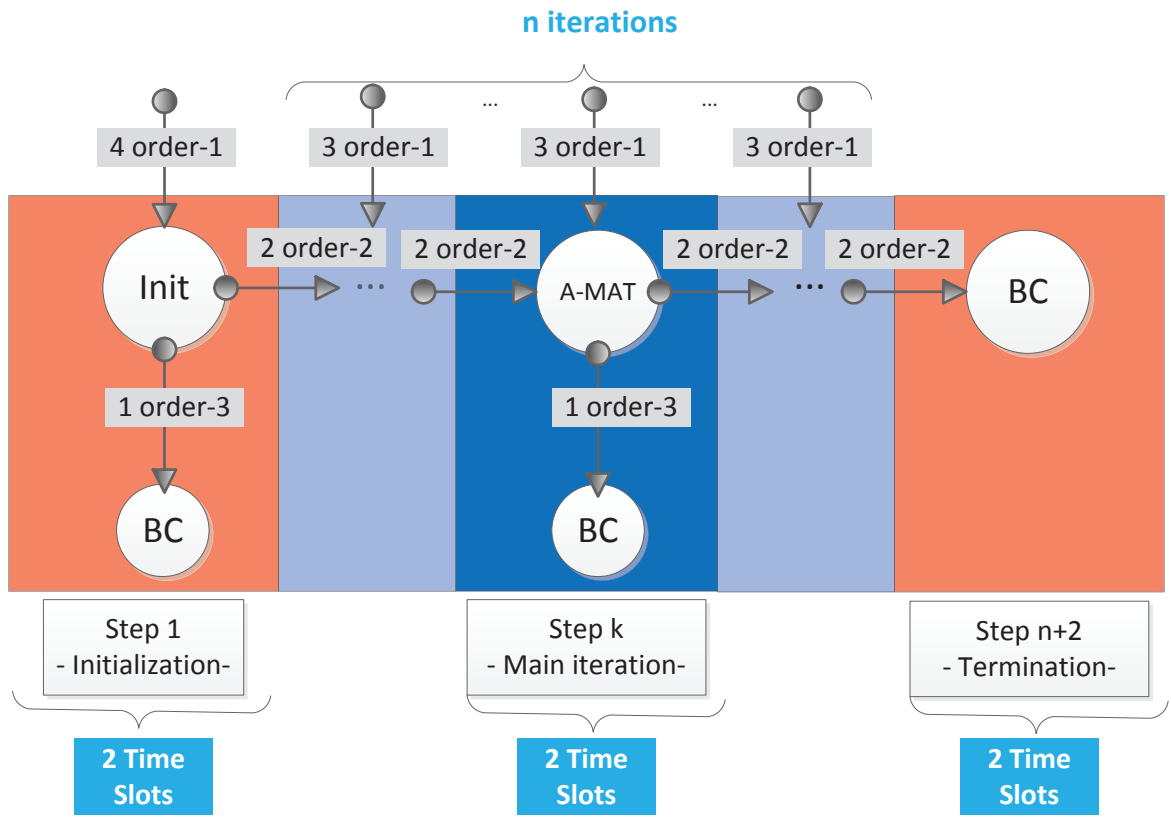}
\caption{Symbolic representation of the $\ALTMAT$ scheme for $K=3$ users.}
\label{ALTMAT_Scheme}
\end{figure}

\subsection{Description of the $\ALTMAT$ Scheme}\label{se:ALTMAT:Kusers}
We will now describe the $\ALTMAT$ scheme for arbitrary values of $K$. The $\ALTMAT$ algorithm can be divided in distinct phases which we denote as \emph{order-$j$ phase}. We will start by presenting the order-$j$ phase before moving to the description of how such phases are combined in the $\ALTMAT$ scheme.

Note that each step should be carried out~$K$ times for the $K$~circular permutations of the users. This is necessary to ensure that every user is transmitted the same number of data symbols. For clarity, we will present the scheme for one particular RX configuration only.

\subsubsection{Order-$j$ Phase}
The order-$j$ phase consists in the simultaneous transmission of messages of order-$j$ and of messages of order-$(K-j)$. We assume wlog that the order-$j$ messages are destined to RX~$1$, RX~$2$, $\ldots$, RX~$j$, while the order-$(K-j)$ messages are destined to the remaining $K-j$~users. We will discuss later on how these messages of order-$j$ and order-$(K-j)$ are obtained. In one time slot, the vector~$\bm{u}_j\in \mathbb{C}^{(K-j+1)\times 1}$ containing the $K-j+1$ data symbols of order-$j$ and the vector~$\bm{u}_{K-j} \in \mathbb{C}^{(j+1)\times 1}$ containing the $j+1$ data symbols of order-$(K-j)$ are transmitted.

Hence, the received signal at RX~$i$ can be written as
\begin{equation}
y_i=\bm{h}_i^{\He}\bm{u}_j+\bm{h}_i^{\He}\bm{u}_{K-j}+z_i. 
\label{eq:ALTMAT_5}
\end{equation}
For $i=1,\ldots,j$, $\bm{h}_i^{\He}\bm{u}_{K-j}$ represents an interfering signal which is desired at RX~$i$ in order to remove the interference. Yet, this is also of interest to RX~$k$ for $k=j+1,\ldots,K$ since it represents an additional equation in $\bm{u}_{K-j}$. Thus, $\bm{h}_i^{\He}\bm{u}_{K-j}$ can be seen as an order-$(K-j+1)$ message. 

Similarly, for $i=j+1,\ldots,K$, $\bm{h}_i^{\He}\bm{u}_{j}$ represents an interfering signal at RX~$i$ but is also of interest to RX~$k$ for $k=1,\ldots,j$. The messages $\bm{h}_i^{\He}\bm{u}_{j}$ for $i=j+1,\ldots,K$ are then of order-$(j+1)$. 

If the $j$~order-$(K-j+1)$~messages and the $K-j$~order-$(j+1)$ messages are transmitted to the RXs who desire these messages, each RX can be seen to have enough interference-free equations to decode its messages. Indeed, the first $j$ (resp. last $K-j$) RXs have received $K-j+1$ (resp. $j+1$) independent equations, which is exactly equal to the number of independent data symbols that they need to decode. The number of time slots~$n_{\mathrm{TS}}$ required for this is then equal to
\begin{equation}
n_{\mathrm{TS}}=\frac{K-j}{\DoF_{j+1}}+\frac{j}{\DoF_{K-j+1}}+1
\label{eq:ALTMAT_6}
\end{equation}
where the addition of a $1$ corresponds to the one time slot used for the transmission in \eqref{eq:ALTMAT_5}. During the $n_{\mathrm{TS}}$ time slots, $K-j+1$ order-$j$ messages and $j+1$ order-$(K-j)$ messages can then be successfully transmitted. From the definition of the DoF, we can then also write $n_{\mathrm{TS}}$ as
\begin{equation}
n_{\mathrm{TS}}=\frac{K-j+1}{\DoF_{j}}+\frac{j+1}{\DoF_{K-j}}.
\label{eq:ALTMAT_7}
\end{equation}
Putting together~\eqref{eq:ALTMAT_7} and \eqref{eq:ALTMAT_6} yields
\begin{equation}
\frac{j+1}{\DoF_{K-j}(K,K)}+\frac{K-j+1}{\DoF_{j}}=\frac{K-j}{\DoF_{j+1}(K,K)}+\frac{j}{\DoF_{K-j+1}}+1.
\label{eq:ALTMAT_8}
\end{equation}
 
\subsubsection{The $\ALTMAT$ Scheme}\label{se:ALTMAT:Kusers:full}

The order-$j$ phase assumes that messages of order-$j$ and messages of order-$(K-j)$ need to be transmitted. We will now show how the order-$j$ phase are combined in the $\ALTMAT$ scheme to allow for the transmission of order-$1$ data symbols.

The proof that the $\ALTMAT$ scheme successfully transmit the data symbols and the derivation of the DoF will be done in the following subsection. We present the $\ALTMAT$ for the case $K$ odd and the modifications required when $K$ is even will be described hereafter.
\begin{itemize}
\item \emph{Step~$1$--Initialization--} The order-$j$ phase is carried out for $j=1,\ldots,(K-1)/2$ but for every phase, the messages of higher order are replaced by the order-$1$ symbols that we aim at transmitting. This is done by choosing arbitrarily any RX among the $j$ destined RXs since the messages are transmitted so as to be decoded at each of the $j$~RXs. This step is spread over $(K-1)/2$ time slots and leads to the creation of messages of order~$j$ for $j=2,\ldots,K$. The number of messages of order-$j$ generated can be obtained from \eqref{eq:ALTMAT_8}. One message of order-$K$ is generated and is directly transmitted via broadcasting.

Note that for clarity a different initialization has been used for~$K=3$ in Subsection\ref{se:ALTMAT:3users}.

\item \emph{Step~$2$ to step~$(n+1)$--Main Iteration--} For every iteration step, all the order-$j$ phases are carried out once for $j=1,\ldots,(K-1)/2$. At the $n$th step, the order-$j$ data symbols being sent are the ones which have been generated during step~$(n-1)$, where the initialization corresponds to step~$0$. The verification that the number of data symbols created matches the number of data symbols needed as inputs will be done in the next subsection.
\item \emph{Step~$n+2$--Termination--} All the data symbols which need to be transmitted are simply broadcasted. This phase can be seen after summation of all the equations given by \eqref{eq:ALTMAT_7} to require $K(K+1)/2-1-(K-1)$ time slots. 
\end{itemize}
If $K$ is even $(K-1)/2$ is replaced by $K/2-1$ and the order-$K/2$ phase is carried out only one time every two steps. The number of time slots used for the termination remains unchanged.
 
\subsection{Sum DoF Achieved}
We will now show that this scheme can indeed be used to achieve the DoF given in Theorem~\ref{thm_ALTMAT}. We start by proving the following lemma.
\begin{lemma}
For every~$j\neq 1,K$, the number of data symbols taken as input in one $\ALTMAT$ iteration is equal to the number of order~$j$ messages generated in such an iteration.
\label{lemma_ALTMAT}
\end{lemma}
\begin{proof}
A detailed proof is provided in Appendix~\ref{app:Lemma_ALTMAT}.
\end{proof}
Using Lemma~\ref{lemma_ALTMAT}, we can compute the DoF achieved by the $\ALTMAT$ scheme by observing how many time slots are used and how many order-$1$ data symbols could be transmitted during those time slots. Let us consider for the moment $K$ to be odd.  
\begin{itemize}
\item \emph{--Initialization--} The initialization step is spread over $(K+1)/2$ time slots and $K(K+1)/2-1$ order-$1$ data symbols are taken as input.
\item \emph{--Main iteration step--} At every time iteration, $K$ order-$1$ data symbols are taken as input and each iteration is spread over $(K+1)/2$ time slots. According to Lemma~\ref{lemma_ALTMAT}, the number of order-$j$ symbols created in every iteration with $j\geq 2$, is the same as the number of order-$j$ messages transmitted. Thus, the DoF of one iteration step is $K/((K+1)/2)=2K/(K+1)$.
\item \emph{--Termination--} The termination step requires $K(K+1)/2-K$ time slots to broadcast all the remaining data symbols.
\end{itemize}

To compute the DoF achieved, it is necessary to take into account the need to consider for every steps the $K$~circular permutations between the users. Hence, the total number of time slots over which the $\ALTMAT$ scheme is spread is equal to 
\begin{equation}
n_{\mathrm{TS}}=K\left(\frac{K+1}{2}+n\frac{K+1}{2}+\frac{K(K+1)-2(K-1)}{2}\right)
\label{eq:ALTMAT_8}
\end{equation}
where the first term in the RHS of \eqref{eq:ALTMAT_7} corresponds to the initialization, the second term to the $n$~main iteration steps, and the third one to the termination step.

In total, the DoF achieved by the $\ALTMAT$ after $n$~steps is then
\begin{equation}
\DoF^{\ALTMAT}_{1}(K,K)=\frac{K\left(K(K+1)-1+n\frac{K+1}{2}\left(\frac{2K}{K+1}\right)\right)}{K\left(\frac{(K+1)}{2}\right)+n\left(\frac{K+1}{2}\right)+\left(\frac{K(K+1)-2(K-1)}{2}\right)}
\label{eq:ALTMAT_9}
\end{equation}
which gives after some basic manipulations the expression in Theorem~\ref{thm_ALTMAT}.

As the number of time slots increases, the $\ALTMAT$ scheme achieves a DoF of $2K/(K+1)$ based on completely outdated CSIT. Although the sum DoF of this new scheme is smaller than the one achieved with $\MAT$, it provides an alternative way to exploit delayed CSIT which will make the exploitation of the prediction obtained from the delayed CSIT more applicable. The $\ALTMAT$ scheme is compared to the $\MAT$ scheme in Fig.~\ref{SumDoF_ALTMAT}.

\begin{figure}
\centering
\includegraphics[width=1\columnwidth]{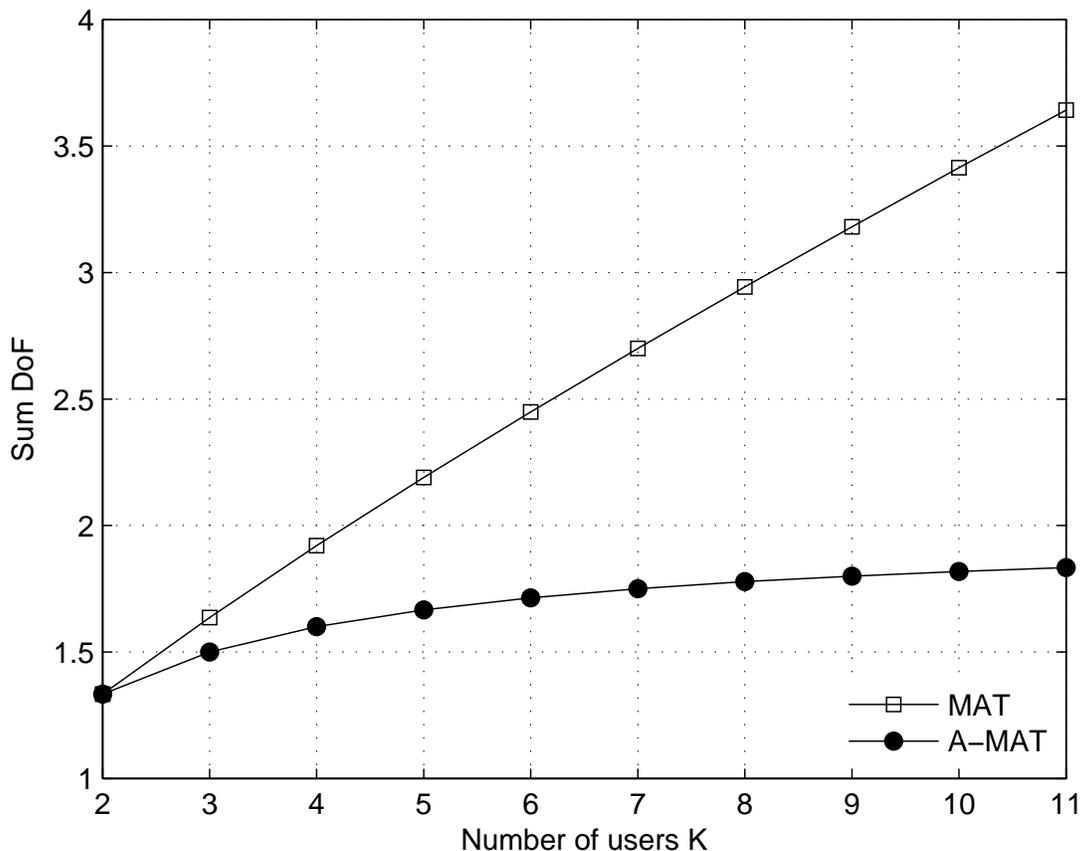}
\caption{Sum DoF in terms of the number of users~$K$.}
\label{SumDoF_ALTMAT}
\end{figure}

\FloatBarrier
\section{The $\KMAT$ Scheme}\label{se:KMAT}
When the CSIT is completely outdated ($\alpha=0$), we will use our new $\ALTMAT$ scheme in place of the MAT scheme. In the other extreme, when $\alpha=1$, ZF is well known to be DoF achieving. Thus, it remains to develop a scheme for the intermediate values of the CSIT quality exponent~$\alpha$. Extending the $\ALTMAT$ scheme to this case will in fact prove to be very easy: The DoF achieved with the modified scheme, which we denote as the $\KMAT$ scheme, will go linearly from the DoF achieved with the $\ALTMAT$ scheme to the DoF achieved with ZF as the CSIT quality exponent~$\alpha$ increases. 

Note that the sum DoF obtained with the outer bound given in Theorem~\ref{thm_out} for a CSIT quality exponent~$\alpha$ is equal to~$(1-\alpha)\DoF^{\MAT}+\alpha\DoF^{\ZF}$ where $\DoF^{\MAT}$ is the DoF achieved with $\MAT$ alignement. Hence, if $\ALTMAT$ were optimal for $\alpha=0$, $\KMAT$ would then be optimal for arbitrary values of $\alpha$. This it the case for $K=2$ where $\ALTMAT$ coincides with the alternative version of $\MAT$. As a consequence, the $\KMAT$ scheme is also optimal. In fact, the $\KMAT$ scheme matches then with the optimal scheme from \cite{Yang2012,Gou2012}.

We will start by describing the different steps of the $\KMAT$ scheme before moving to the analysis of the DoF achieved. 
\subsection{Description of the $\KMAT$ Scheme}\label{se:KMAT:KMAT}
We will show how the order-$j$ phase of the $\ALTMAT$ scheme is modified to exploit the correlation between the delayed CSIT and the instantaneous channel. The full $\KMAT$ scheme follows then trivially from the description of the $\ALTMAT$ scheme in Section~\ref{se:ALTMAT}.

We assume wlog that the order-$j$ symbols are destined to the first $j$~TXs and the order-$(K-j)$ symbols to the $K-j$ last RXs.
\begin{itemize}
\item \emph{Direct Transmission: } 
\paragraph{The $\ALTMAT$ Data Symbols} According to the $\ALTMAT$ scheme, the TX transmit $K-j+1$ order-$j$ messages and $j+1$ order-$(K-j)$ messages. Yet, the data symbols are this time precoded. The $i$th order-$j$ data symbol is precoded to form the vector~$\bm{a}^{(j)}_i\in \mathbb{C}^{M\times 1}$ while the $k$th order-$(K-j)$ data symbol is precoded as the vector~$\bm{a}^{(K-j)}_k\in \mathbb{C}^{M\times 1}$. The vector $\bm{a}^{(j)}_1$ is chosen to ZF the interference to the $K-j$ last RXs, i.e., such that 
\begin{equation}
\forall k=j+1,\ldots,K,~~\hat{\bm{h}}_k^{\He}\bm{a}^{(j)}_1=0.
\label{eq:KMAT_1}
\end{equation}
The remaining $K-j$ precoded data symbols are chosen such that $\forall k<i,(\bm{a}^{(j)}_k)^{\He}\bm{a}^{(j)}_i=0$\footnote{Note that this is solely done to ensure that all the precoded data symbols are linearly independent and span a subspace of dimension~$K-j+1$.}. Similarly, $\bm{a}^{(K-j)}_1$ is chosen such that 
\begin{equation}
\forall k=1,\ldots,j,~~\hat{\bm{h}}_k^{\He}\bm{a}^{(K-j)}_1=0
\label{eq:KMAT_2}
\end{equation}
and the remaining $j$ beamformers such that $\forall k<i, (\bm{a}^{(K-j)}_k)^{\He}\bm{a}^{(K-j)}_i=0$.

The power is allocated to these precoded data symbols as follows.
\begin{equation}
\begin{cases}
k=1,&\norm{\bm{a}^{(j)}_k}^2=\left[\frac{1}{2}\left(P-P^{\alpha}\right)-\frac{1}{2}\frac{K-j}{K-j+1}P^{1-\alpha}\right]^{+},\\
\forall k=2,\ldots,K-j+1&\norm{\bm{a}^{(j)}_k}^2=\frac{1}{2}\frac{1}{K-j+1}P^{1-\alpha}
\end{cases}
\label{eq:KMAT_3}
\end{equation}
and similarly
\begin{equation}
\begin{cases}
k=1,&\norm{\bm{a}^{(K-j)}_k}^2=\left[\frac{1}{2}\left(P-P^{\alpha}\right)-\frac{1}{2}\frac{j}{j+1}P^{1-\alpha}\right]^{+},\\
\forall k=2,\ldots,j+1&\norm{\bm{a}^{(K-j)}_k}^2=\frac{1}{2}\frac{1}{j+1}P^{1-\alpha}
\end{cases} 
\label{eq:KMAT_4}
\end{equation}
The reason for this particular power allocation will become clear in the decoding part of the scheme. Every data symbol is sent with the rate~$(1-\alpha)\log(P)$.  

\paragraph{The ZF Data Symbols} In addition to these data symbols, we will transmit at the same time via conventional ZF one data symbol~$s_j$ to RX~$j$ (i.e an order-$1$ data symbol) for every RX~$j$. Hence, the data symbol $s_j$ is precoded to obtain~$\bm{u}_j\in \mathbb{C}^{M\times 1}$ such that 
\begin{equation}
\forall k\neq j,\hat{\bm{h}}_k^{\He}\bm{u}_j=0.
\label{eq:KMAT_5}
\end{equation}
The power is allocated to verify that $\forall i, \E[\norm{\bm{u}_i}^2]=P^{\alpha}/K$ and each data symbol is sent with the rate $\alpha\log(P)$.

The received signal at RX~$k$ then reads as
\begin{equation}
\begin{cases}
k\leq j, &y_k=\underbrace{\bm{h}^{\He}_k\bm{a}^{(j)}_1}_{\sim P}+\underbrace{\sum_{i=2}^{K-j+1}\bm{h}^{\He}_k\bm{a}^{(j)}_i}_{\sim P^{1-\alpha}}+\underbrace{\sum_{i=1}^{j+1}\bm{h}^{\He}_k\bm{a}^{(K-j)}_i}_{\sim P^{1-\alpha}}+\underbrace{\sum_{i=1}^K\bm{h}^{\He}_k\bm{u}_i}_{\sim P^{\alpha}}+z_k\\
k\geq j+1, &y_k=\underbrace{\bm{h}^{\He}_k\bm{a}^{(K-j)}_1}_{\sim P}+\underbrace{\sum_{i=2}^{j+1}\bm{h}^{\He}_k\bm{a}^{(K-j)}_i}_{\sim P^{1-\alpha}}+\underbrace{\sum_{i=1}^{K-j+1}\bm{h}^{\He}_k\bm{a}^{(j)}_i}_{\sim P^{1-\alpha}}+\underbrace{\sum_{i=1}^K\bm{h}^{\He}_k\bm{u}_i}_{\sim P^{\alpha}}+z_k
\end{cases}
\label{eq:KMAT_6}
\end{equation}
Note that the interferences from $\bm{a}^{(j)}_1$ and $\bm{a}^{(K-j)}_1$ have been attenuated by $P^{-\alpha}$ following the ZF with respect to the imperfect channel estimates.
\item \emph{Creation of the $\ALTMAT$ Order-$j+1$ Data Symbols:} Considering the received signal scaling in $P^{\alpha}$ as noise and omitting the power scaling of the received signals, we have obtained the same received signals as in the $\ALTMAT$ scheme described in Section~\ref{se:ALTMAT}. Hence, the interference $\sum_{i}\bm{h}^{\He}_k\bm{a}^{(K-j)}_i$ for $k\leq j$ is needed to remove the interference at RX~$k$ but forms also a desired equation for the last $K-j$ users. Thus, it can be seen as an order-$(K-j+1)$ message. Similarly, the interference $\sum_{i}\bm{h}^{\He}_k\bm{a}^{(j)}_i$ for $k\geq j+1$ is needed by the first $j$ RXs and by RX~$k$, and is hence an order-$(j+1)$ message.

All the ``equations" which have to be retransmitted have a power scaling in $P^{1-\alpha}$. Hence, we can use the well known result that quantizing them with $(1-\alpha)\log(P)$~bits leads to a distorsion scaling in $P^0$\cite{cover2006}, which is negligible in terms of DoF. 

The data symbols of order-$j$ and order-$(K-j)$ taken as input have a rate of $(1-\alpha)\log(P)$ and this is also the case of the new messages created. As a consequence, the $\ALTMAT$ algorithm can proceed with the transmission of the quantized equations as the order-$(j+1)$ and order-$(K-j+1)$ messages for the next iteration of the $\ALTMAT$ scheme.
\item \emph{Successive decoding:} We now consider that the modified $\ALTMAT$ has reached its end. Let us first consider RX~$k$ for $k\leq j$. This RX has received $K-j$ equations relative to its order-$j$ symbols and was also able to remove the interference received. Hence, it has in total $K-j+1$ equations having each a SNR scaling in $P^{1-\alpha}$. Consequently, RX~$k$ can decode all the desired precoded data symbols~$\bm{a}^{(j)}_i$ for all $i$.
\item \emph{Successive decoding:} We now consider that the modified $\ALTMAT$ has reached its end. Let us first consider RX~$k$ for $k\leq j$. This RX has received $K-j$ equations relative to its order-$j$ symbols and was also able to remove the interference received. Hence, it has in total $K-j+1$ equations having each a SNR scaling in $P^{1-\alpha}$. Consequently, RX~$k$ can decode all the desired precoded data symbols~$\bm{a}^{(j)}_i$ for all $i$.

The data symbols of order-$j$ being decoded, they can be subtracted from the received signal. Since the interference have also been subtracted, the received signal at RX~$k$ reads then as
\begin{equation} 
y_k=\underbrace{\bm{h}^{\He}_k\bm{u}_k}_{\sim P^{\alpha}}+\underbrace{\sum_{i=1,i\neq k}^K\bm{h}^{\He}_k\bm{u}_i}_{\sim P^{0}}+z_k.
\label{eq:KMAT_7}
\end{equation}
The interference term in \eqref{eq:KMAT_7} is drawn in the noise due to the attenuation by~$P^{-\alpha}$ from the ZF precoding. As a consequence, the precoded symbol~$\bm{u}_k$ is received at RX~$k$ with a SNR scaling as~$P^{\alpha}$ and can be decoded.

The same analysis can be carried out for RX~$k$ with $k\geq j$.
\end{itemize}

\subsection{Degrees of Freedom Analysis}\label{se:KMAT:DoF}

 \begin{figure}[htp!] 
\centering
\includegraphics[width=1\columnwidth]{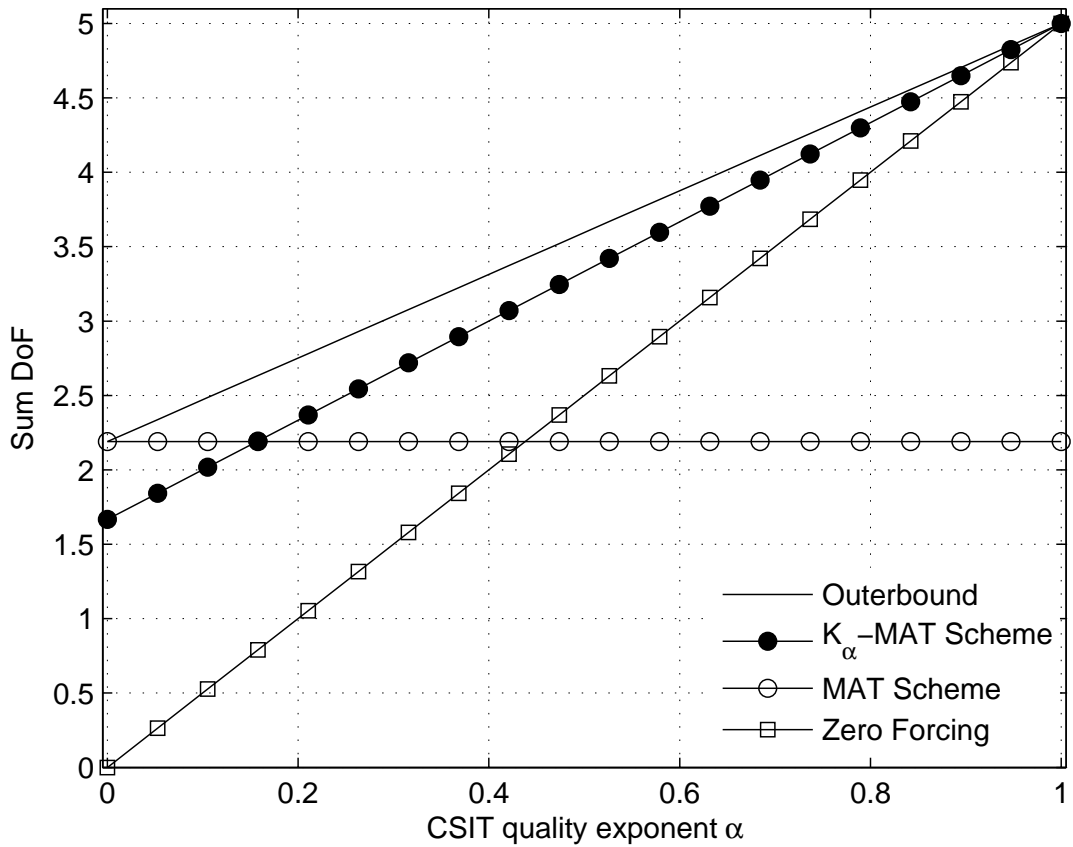}
\caption{Sum DoF for $K=5$~users in terms of the CSIT quality exponent~$\alpha$.}
\label{SumDoF_5Users}
\end{figure}
From the description of the algorithm, the DoF expression from Theorem~\ref{thm_KMAT} is easily derived as follows. The $\ALTMAT$ scheme has been used to transmit data symbol of rate~$(1-\alpha)\log(P)$ while at \emph{every time slot} of this scheme, one data symbol has been transmitted to every user via ZF with a rate equal to~$\alpha \log(P)$. Hence, the DoF given in Theorem~\ref{thm_KMAT} can be achieved.

In Fig.~\ref{SumDoF_5Users}, we represent the sum DoF achieved with the~$\KMAT$ scheme. Although the $\MAT$ scheme is optimal when $\alpha=0$ and the CSIT is completely outdated, the $\ALTMAT$ scheme becomes more efficient as the CSIT quality exponent increases. The $\KMAT$ scheme coincides with ZF when the CSIT is accurate enough ($\alpha=1$) and is otherwise more performing. Hence, it can be seen as a robust version of ZF with respect to the delay in the CSIT.

Furthermore, we show in Fig.~\ref{sum_DoF_KMAT_alpha05} the DoF achieved in terms of the number of users with the CSIT quality exponent $\alpha=0.5$. It can be seen that the $\KMAT$ scheme outperforms in that case both ZF and $\MAT$.
\begin{figure}[htp!] 
\centering
\includegraphics[width=1\columnwidth]{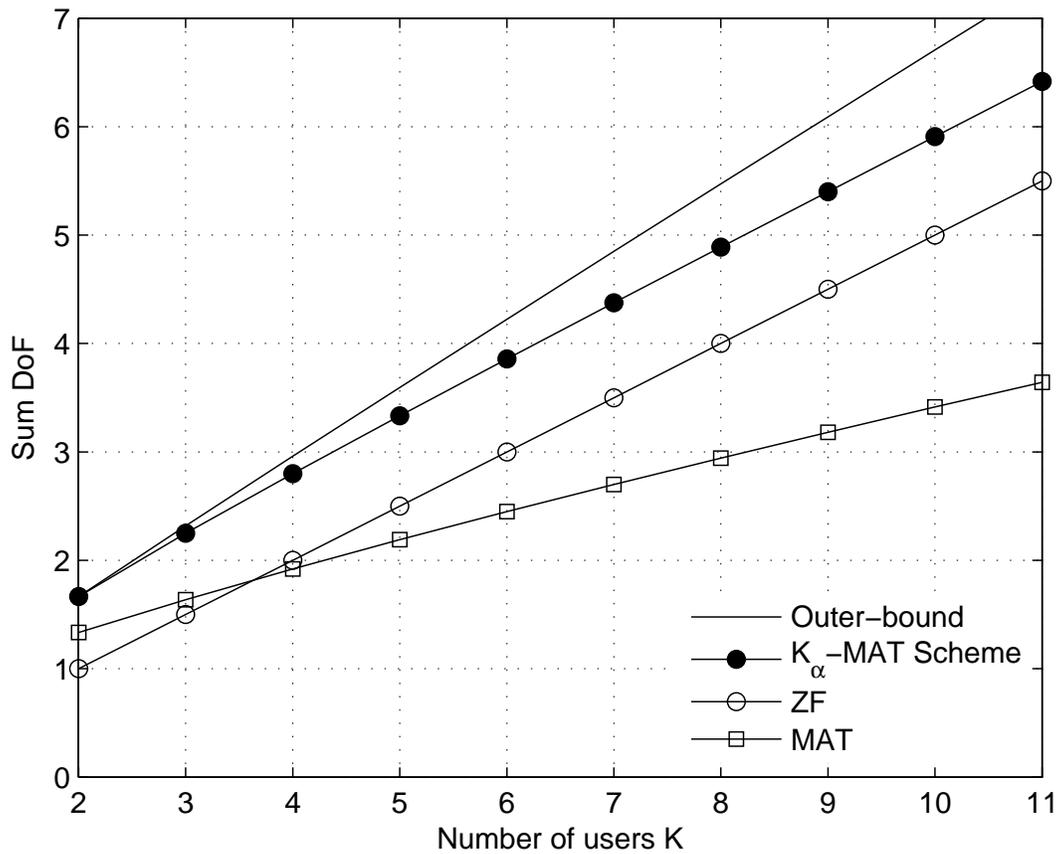}
\caption{Sum DoF in terms of the number of users~$K$ for the CSIT quality exponent~$\alpha=0.5$.}
\label{sum_DoF_KMAT_alpha05}
\end{figure}

\FloatBarrier 
\newpage
\section{Proof of the Outer Bound}\label{se:Outerbound}

To obtain the outer bound, we adopt a genie-aided upper bounding technique inspired from\cite{Vaze2011a,Yang2012}. We provide to RX~$i$ the side information of the RX~$j$'s message $W_j$ as well as the received signal $y_j(t'), \forall t'\leq t$ for $j = i+1,\cdots,K$. We consider that all the $K$~users are active (i.e., have a positive DoF) because the approach trivially extends by replacing~$K$ with any number~$p$ of active users such that~$1\leq p\leq K$. Recall that all the RXs have access after a given delay to the perfect CSI~${\mathbf{H}}(t)$ as well as the imperfect CSI $\hat{\mathbf{H}}(t)$. Since the decoding of the signal received at time $t$ is done solely once the RX has received the CSI relative to time $t$, it means that we can consider that the RXs have access to the CSI instantaneously. We further define for ease of notation $W_{[i:j]}\triangleq \{W_{i},W_{i+1},\cdots,W_{j}\}$, $\mathbf{Y}_{[i:j]}(t) \triangleq \{y_i(t), y_{i+1}(t),\cdots,y_j(t)\}$, $\mathbf{H}_{[i:j]}(t) \triangleq [\bm{h}_i(t) , \bm{h}_{i+1}(t) , \cdots,\bm{h}_j(t)]^\He$, where $j\ge i$, and $\mathbf{Y}_{[i:j]}^t \triangleq \{\mathbf{Y}_{[i:j]}(m)\}_{m=1}^t$.

From Fano's inequality, it follows for arbitrary~$\varepsilon_n>0$,
\begin{align}
  n(R_k-\varepsilon_n) &\leq I(W_k;W_{[k+1:K]},\mathbf{Y}_{[k:K]}^n | {\mathbf{H}}^n, \hat{\mathbf{H}}^n )\\
  &= I(W_k;\mathbf{Y}_{[k:K]}^n | W_{[k+1:K]}, {\mathbf{H}}^n, \hat{\mathbf{H}}^n )\\
  &= \sum_{t=1}^n I(W_k;\mathbf{Y}_{[k:K]}(t) | W_{[k+1:K]},\mathbf{Y}_{[k:K]}^{t-1},{\mathbf{H}}^n, \hat{\mathbf{H}}^n )\\
  &= \sum_{t=1}^n \left(h(\mathbf{Y}_{[k:K]}(t) | W_{[k+1:K]},\mathbf{Y}_{[k:K]}^{t-1},{\mathbf{H}}^t, \hat{\mathbf{H}}^t ) - h(\mathbf{Y}_{[k:K]}(t) | W_{[k:K]},\mathbf{Y}_{[k:K]}^{t-1}, {\mathbf{H}}^t, \hat{\mathbf{H}}^t ) \right)\\
  &= \sum_{t=1}^n \left(h(\mathbf{Y}_{[k:K]}(t) | \mathcal{U}_k(t),\mathbf{H}(t)) - h(\mathbf{Y}_{[k:K]}(t) | W_k, \mathcal{U}_k(t),\mathbf{H}(t)) \right)
\end{align}
where we have defined $\mathcal{U}_k(t) \triangleq \{W_{[k+1:K]},\mathbf{Y}_{[k:K]}^{t-1},{\mathbf{H}}^{t-1}, \hat{\mathbf{H}}^{t}\}$.
Thus, the weighted sum rate can be bounded for arbitrary nonzero natural number~$N_k,k=1,\ldots,K$ as
\begin{align}
  &\sum_{k=1}^K \frac{n(R_k-\varepsilon_n)}{N_k} \notag\\
   &\leq \sum_{t=1}^n \sum_{k=1}^K \frac{1}{N_k} h(\mathbf{Y}_{[k:K]}(t) | \mathcal{U}_k(t),\mathbf{H}(t)) - \sum_{t=1}^n \sum_{k=1}^K \frac{1}{N_k} h(\mathbf{Y}_{[k:K]}(t) | W_k, \mathcal{U}_k(t),\mathbf{H}(t))\\
  &= \sum_{t=1}^n \sum_{k=1}^{K-1} \left\{\frac{1}{N_k} h(\mathbf{Y}_{[k:K]}(t) | \mathcal{U}_k(t),\mathbf{H}(t)) - \frac{1}{N_{k+1}} h(\mathbf{Y}_{[k+1:K]}(t) | W_{k+1},\mathcal{U}_{k+1}(t),\mathbf{H}(t)) \right\}\notag\\
  &~~~~~~~~~~~~~~~~~~~~~~~~~~+ \frac{1}{N_K} h(y_{K}(t) | \mathcal{U}_K(t),\mathbf{H}(t)) - \frac{1}{N_1} h(\mathbf{Y}_{[1:K]}(t) | W_{1},\mathcal{U}_1(t),\mathbf{H}(t))\\
  &\leq\! \sum_{t=1}^n \!\sum_{k=1}^{K-1} \!\left\{\frac{1}{N_k} h(\mathbf{Y}_{[k:K]}(t) | \mathcal{U}_k(t),\mathbf{H}(t))\! - \!\frac{1}{N_{k+1}} h(\mathbf{Y}_{[k+1:K]}(t) |W_{k+1},\mathcal{U}_{k+1}(t),\mathbf{H}(t),\mathbf{Y}_k^{t-1}) \!\!\right\}\!\notag\\
  &~~~~~~~~~~~~~~~~~~~~~~~~~~+ n \log P+n \cdot O(1)\\
  &= \sum_{t=1}^n \sum_{k=1}^{K-1} \left\{\frac{1}{N_k} h(\mathbf{Y}_{[k:K]}(t) | \mathcal{U}_k(t),\mathbf{H}(t)) - \frac{1}{N_{k+1}} h(\mathbf{Y}_{[k+1:K]}(t) | \mathcal{U}_k(t),\mathbf{H}(t)) \right\}\notag\\
  &~~~~~~~~~~~~~~~~~~~~~~~~~~+ n \log P+n \cdot O(1)
\label{eq:proof_out_1}
\end{align} 
Let us focus on one of the differences of entropy in the summation. We can apply the same calculation as in the proof of the outerbound in \cite{Yang2012}. Firstly, we set~$\forall k, N_k=K-k+1$ to write 
\begin{align}
  &\frac{1}{K-k+1} h(\mathbf{Y}_{[k:K]}(t) | \mathcal{U}_k(t),\mathbf{H}(t)) - \frac{1}{K-k} h(\mathbf{Y}_{[k+1:K]}(t) | \mathcal{U}_k(t),\mathbf{H}(t))\notag\\
&\leq \max_{p(\mathcal{U}_k(t)),p(\xv(t)|\mathcal{U}_k(t))} \left(\frac{h(\mathbf{Y}_{[k:K]}(t) | \mathcal{U}_k(t),\mathbf{H}(t))}{K-k+1}  - \frac{h(\mathbf{Y}_{[k+1:K]}(t) | \mathcal{U}_k(t),\mathbf{H}(t))}{K-k} \right)\label{eq:proof_out_3_1}\\
&\leq \max_{\mathcal{U}_k(t))}\E_{\mathcal{U}_k(t)}\max_{p(\xv(t)|\mathcal{U}_k(t))} \left(\frac{h(\mathbf{Y}_{[k:K]}(t) | \mathcal{U}_k(t),\mathbf{H}(t))}{K-k+1}  - \frac{h(\mathbf{Y}_{[k+1:K]}(t) | \mathcal{U}_k(t),\mathbf{H}(t))}{K-k} \right)\label{eq:proof_out_3_2}\\
&= \max_{\mathcal{U}_k(t))}\E_{\mathcal{U}_k(t)}\max_{p(\xv(t)|\mathcal{U}_k(t))}\E_{\mathbf{H}(t)|\mathcal{U}_k(t)} \left(\frac{h(\mathbf{Y}_{[k:K]}(t) | \mathcal{U}_k(t),\mathbf{H}(t)) }{K-k+1} - \frac{h(\mathbf{Y}_{[k+1:K]}(t) | \mathcal{U}_k(t),\mathbf{H}(t))}{K-k} \right)\label{eq:proof_out_3_2}\\
&= \max_{\mathcal{U}_k(t))}\E_{\mathcal{U}_k(t)}\max_{p(\xv(t)|\mathcal{U}_k(t))}\E_{\mathbf{H}(t)|\hat{\mathbf{H}}(t)} \left(\frac{h( \mathbf{H}_{[k:K]}(t) \xv(t)+\bm{z}_{[k:K]}(t) | \mathcal{U}_k(t)) }{K-k+1} \right.\notag\\
	&\left.~~~~~~~~~~~~~~~~~~~~~~~~~~~~~~~~~~~~~~~~~~~~~~~~~~~ -\frac{h( \mathbf{H}_{[k+1:K]}(t) \xv(t)+\bm{z}_{[k+1:K]}(t) | \mathcal{U}_k(t))}{K-k} \right)\label{eq:proof_out_3_3}\\
&= \max_{\mathcal{U}_k(t))}\E_{\mathcal{U}_k(t)}\max_{\substack{\mathbf{C}\succeq 0\\\tr(\mathbf{C})\leq P}}\max_{\substack{p(\xv(t)|\mathcal{U}_k(t))\\\mathrm{cov}(\xv(t)|\mathcal{U}_k(t))\preceq \mathbf{C}}}\E_{\mathbf{H}(t)|\hat{\mathbf{H}}(t)} \left(\frac{h( \mathbf{H}_{[k:K]}(t) \xv(t)+\bm{z}_{[k:K]}(t) | \mathcal{U}_k(t)) }{K-k+1} \right.\notag\\
	&\left.~~~~~~~~~~~~~~~~~~~~~~~~~~~~~~~~~~~~~~~~~~~~~~~~~~~ - \frac{ h( \mathbf{H}_{[k+1:K]}(t) \xv(t)+\bm{z}_{[k+1:K]}(t) | \mathcal{U}_k(t))}{K-k}\right)\label{eq:proof_out_3_4}
\end{align}
where \eqref{eq:proof_out_3_2} is obtained because maximizing inside the expectation leads to an upper bound and \eqref{eq:proof_out_3_4} follows from splitting the constraint on the distribution in two constraints.


We can now apply the Extremal Inequality from \cite[Theorem~$8$]{Liu2007}. This is possible because $\xv(t)$ is independent of $\mathbf{H}(t)$ (and of the noise) conditioned on the channel estimate~$\hat{\mathbf{H}}(t)$. The multiplication by the channel matrices (not present in the original theorem) is taking care of by inverting the channel after having regularized it, and letting then the regularization tend to zero \cite{Weingarten2006}.

It follows from that result that the optimal vector~$\xv(t)$ is Gaussian distributed. We define then the covariance matrix~$\mathbf{K}_{\xv}(t) \triangleq \E \{\xv(t)\xv^\He(t) | \mathcal{U}_k(t)\}$ and write
\begin{align}
  &\frac{1}{K-k+1} h(\mathbf{Y}_{[k:K]}(t) | \mathcal{U}_k(t),\mathbf{H}(t)) - \frac{1}{K-k} h(\mathbf{Y}_{[k+1:K]}(t) | \mathcal{U}_k(t),\mathbf{H}(t))\notag\\		
&\leq \max_{\mathcal{U}_k(t))}\E_{\mathcal{U}_k(t)}\max_{\substack{\mathbf{C}\succeq 0\\\tr(\mathbf{C})\leq P}}\max_{\mathbf{K}_{\xv}(t)\preceq \mathbf{C}}\E_{\mathbf{H}(t)|\hat{\mathbf{H}}(t)} \left(\frac{1}{K\!-\!k\!+\!1}\!\log \det (\mathbf{I}_{K-k+1} \!+ \!\mathbf{H}_{[k:K]}(t) \mathbf{K}_{\xv}(t) \mathbf{H}_{[k:K]}^{\He}(t) ) \!\right. \!\notag\\
	&\left.~~~~~~~~~~~~~~~~~~~~~~~~~~~~~~ - \frac{1}{K-k} \log \det (\mathbf{I}_{K-k}+\mathbf{H}_{[k+1:K]}(t) \mathbf{K}_{\xv}(t) \mathbf{H}_{[k+1:K]}^\He(t))\right) \label{eq:proof_out_4_1}\\
	&= \max_{\mathcal{U}_k(t))}\E_{\mathcal{U}_k(t)}\max_{\substack{\mathbf{C}\succeq 0\\\tr(\mathbf{C})\leq P}}\E_{\mathbf{H}(t)|\hat{\mathbf{H}}(t)} \left(\frac{1}{K\!-\!k\!+\!1}\!\log \det (\mathbf{I}_{K-k+1} \!+ \!\mathbf{H}_{[k:K]}(t) \mathbf{K}^*(t) \mathbf{H}_{[k:K]}^{\He}(t) ) \!\right. \!\notag\\
	&\left.~~~~~~~~~~~~~~~~~~~~~~~~~~~~~~ - \frac{1}{K-k} \log \det (\mathbf{I}_{K-k}+\mathbf{H}_{[k+1:K]}(t) \mathbf{K}^*(t) \mathbf{H}_{[k+1:K]}^\He(t))\right) \label{eq:proof_out_4_2}\\
		&\leq\max_{\mathcal{U}_k(t))}\E_{\mathcal{U}_k(t)}\max_{\substack{\mathbf{C}\succeq 0\\\tr(\mathbf{C})\leq P}}\E_{\mathbf{H}(t)|\hat{\mathbf{H}}(t)} \left(\frac{1}{K\!-\!k\!+\!1}\!\log \det (\mathbf{I}_{K-k+1} \!+ \!\mathbf{H}_{[k:K]}(t) \mathbf{C}(t) \mathbf{H}_{[k:K]}^{\He}(t) ) \!\right. \!\notag\\
	&\left.~~~~~~~~~~~~~~~~~~~~~~~~~~~~~~ - \frac{1}{K-k} \log \det (\mathbf{I}_{K-k}+\mathbf{H}_{[k+1:K]}(t) \mathbf{C}(t) \mathbf{H}_{[k+1:K]}^\He(t))\right) \label{eq:proof_out_4_3}\\
   &\stackrel{a}{\leq} \frac{1}{K-k+1} \alpha \log P + O(1)\label{eq:proof_out_4_4}
\end{align}
where we have defined $\mathbf{K}^*$ as the covariance matrix solution of the inner maximization in \eqref{eq:proof_out_4_1}. Inequality $a$ is a consequence of the following lemma which is proven in Appendix~\ref{app:lemma_out}:
\begin{lemma} 
Let us consider two $N_k \times M$ ($k=1,2$) random matrices $\mathbf{H}_k = \hat{\mathbf{H}}_k+\tilde{\mathbf{H}}_k$, where $\tilde{\mathbf{H}}_k$ has its entries distributed as i.i.d. $\CN(0,\sigma^2)$ and independent of $\hat{\mathbf{H}}_k$. Given any $\mathbf{K} \succeq 0$ with eigenvalues $\lambda_1 \geq \cdots \geq \lambda_{M} \geq 0$, and $M \geq N_1 \geq N_2$, if $\sigma^2$ tends to zero, then 
\begin{align}
 &\frac{1}{N_1}\E_{\tilde{\mathbf{H}}_1} \log \det (\Id_{N_1}+\mathbf{H}_1 \mathbf{K} \mathbf{H}_1^\He) - \frac{1}{N_2} \E_{\tilde{\mathbf{H}}_2} \log \det (\Id_{N_2}+\mathbf{H}_2 \mathbf{K} \mathbf{H}_2^\He) \leq - \frac{N_1-N_2}{N_1} \log (\sigma^2) + O(1).
\label{eq:Lemma_Out}
\end{align}
\label{Lemma_Out}
\end{lemma} 
Using \eqref{eq:proof_out_4_4} in \eqref{eq:proof_out_1} with $N_k=K-k+1$, it follows that
\begin{align}
\sum_{k=1}^K \frac{n(R_k-\varepsilon_n)}{K-k+1}\leq \sum_{t=1}^n \sum_{k=1}^{K-1}\frac{1}{K-k+1} \alpha \log P+ n \log P+n \cdot O(1).
\label{eq:proof_out_5} 
\end{align}
Dividing by~$n\log(P)$, considering arbitrarily long codewords, and letting $P$ tend to infinity gives
\begin{align}
  \sum_{k=1}^K \frac{d_k}{K-k+1} &\leq 1+\alpha\sum_{k=1}^{K-1}\frac{1}{K-k+1}\\
	&=1+ \alpha\sum_{k=2}^{K} \frac{1}{k}.
	\label{eq:proof_out_6} 
\end{align}

By permutation of the users and variation of the number of active users, all the outer bounds can be obtained. This concludes the proof.

\section*{Acknowledgment}
Helpful discussions with Sheng Yang (Supelec) and Mari Kobayashi (Supelec) are gratefully acknowledged.

\section{Conclusion}\label{se:Conclusion}
In this work, considering a $K$-user MISO BC, a new transmission scheme has been developed to exploit at the same time the principle behind the $\MAT$ alignment based on delayed CSIT and ZF of the interference. The novel $\KMAT$ scheme is more robust than ZF to the channel estimates being received with some delay and coincides with ZF when the CSIT received is accurate enough. Furthermore, over a wide range of values taken by the CSIT quality exponent~$\alpha$, the $\KMAT$ scheme outperforms both $\MAT$ and ZF. This makes such approach a strong candidate to improve the robusteness to CSI feedback delays of the transmission scheme. In addition, an outer-bound DoF region has been derived. How to reduce the gap between the outer and the inner bound is an interesting open problem for futur research. Furthermore, the $\MAT$ alignment scheme from Maddah-Ali and Tse is very recent and is expected to have applications in many more settings and to have a strong potential for further improvements.


\clearpage
\appendices
\section{Proof of Lemma~\ref{lemma_ALTMAT}}\label{app:Lemma_ALTMAT}
\begin{proof}
Let us recall first for the sake of clarity the DoF expression for the order-$j$ phase
\begin{equation}
\frac{j+1}{\DoF_{K-j}}+\frac{K-j+1}{\DoF_{j}}=\frac{K-j}{\DoF_{j+1}}+\frac{j}{\DoF_{K-j+1}}+1.
\label{eq:ALTMAT_proof_1}
\end{equation}
Rewriting this expression for the order-$j+1$ phase gives
\begin{equation}
\frac{j+2}{\DoF_{K-j-1}}+\frac{K-j}{\DoF_{j+1}}=\frac{K-j-1}{\DoF_{j+2}}+\frac{j+1}{\DoF_{K-j}}+1.
\label{eq:ALTMAT_proof_2}
\end{equation}
and for the order-$j-1$ phase 
\begin{equation}
\frac{j}{\DoF_{K-j+1}}+\frac{K-j+2}{\DoF_{j-1}}=\frac{K-j+1}{\DoF_{j}}+\frac{j-1}{\DoF_{K-j+2}}+1.
\label{eq:ALTMAT_proof_3}
\end{equation}

Adding \eqref{eq:ALTMAT_proof_1} and \eqref{eq:ALTMAT_proof_2}, the first term of the Left-Hand Side (LHS) of \eqref{eq:ALTMAT_proof_1} simplifies with the second term of the right-hand side (RHS) in \eqref{eq:ALTMAT_proof_2} while the first term of the RHS of \eqref{eq:ALTMAT_proof_1} simplifies with the second term of the LHS of \eqref{eq:ALTMAT_proof_2}. Similarly, adding \eqref{eq:ALTMAT_proof_1} and \eqref{eq:ALTMAT_proof_3}, leads to the simplification of the second term of the LHS and the second term of the RHS in \eqref{eq:ALTMAT_proof_1} with their counterpart in \eqref{eq:ALTMAT_proof_3}. 

As a consequence, adding the equations obtained from phase~$1$ to phase~$k$ yields
\begin{equation}
\frac{K}{\DoF_{1}}+\frac{k+1}{\DoF_{K-k}}=\frac{K-k}{\DoF_{k+1}}+\frac{1}{\DoF_{K}}+k.
\label{eq:ALTMAT_proof_4}
\end{equation}
We now differentiate between the two cases $K$ even and $K$ odd.
\begin{itemize}
\item If $K$ is odd, then choosing $k=(K-1)/2$ in \eqref{eq:ALTMAT_proof_4} gives
\begin{equation}
\frac{K}{\DoF_{1}}=\frac{1}{\DoF_{K}}+\frac{K-1}{2}.
\label{eq:ALTMAT_proof_5}
\end{equation}
because it holds in that case that~$K-k=k+1$ such that two terms simplify in \eqref{eq:ALTMAT_proof_4}. The proof concludes by using that~$\DoF_{K}(K,K)=1$.
\item If $K$ is even, writing \eqref{eq:ALTMAT_proof_4} with $k=K/2-1$ gives
\begin{equation}
\frac{K}{\DoF_{1}}+\frac{\frac{K}{2}}{\DoF_{\frac{K}{2}+1}}=\frac{\frac{K}{2}+1}{\DoF_{\frac{K}{2}}}+\frac{1}{\DoF_{K}}+\frac{K}{2}-1.
\label{eq:ALTMAT_proof_6}
\end{equation}
We proceed by writing the DoF expression \eqref{eq:ALTMAT_proof_1} for the order-$K/2$ phase which gives 
\begin{equation}
\frac{K+2}{\DoF_{\frac{K}{2}}}=\frac{K}{\DoF_{\frac{K}{2}+1}}+1.
\label{eq:ALTMAT_proof_7}
\end{equation}
Adding one half of \eqref{eq:ALTMAT_proof_7} to \eqref{eq:ALTMAT_proof_6} gives \eqref{eq:ALTMAT_proof_5}.
\end{itemize} 
The result follows directly from \eqref{eq:ALTMAT_proof_5} since the expression relative to the symbol of order-$j$ for $j\neq 1,K$ have been simplified.
\end{proof}

\section{Proof of Lemma~\ref{Lemma_Out}}\label{app:lemma_out}
We will proceed by bounding first separately each term of \eqref{eq:Lemma_Out}.
\begin{itemize}
\item Let us consider first the second term which we should lower bound. Recall that we consider two $N_k \times M$ ($k=1,2$) random matrices $\mathbf{H}_k = \hat{\mathbf{H}}_k+\tilde{\mathbf{H}}_k$, where $\tilde{\mathbf{H}}_k$ has its entries distributed as i.i.d. $\CN(0,\sigma^2)$ and independent of $\hat{\mathbf{H}}_k$ and a matrix $\mathbf{K} \succeq 0$ of size~$M\times M$ with eigenvalues $\lambda_1 \geq \cdots \geq \lambda_{M} \geq 0$ such that $M \geq N_1 \geq N_2$. We also define the EigenValue Decomposition (EVD) of the positive semi-definite matrix~$\mathbf{K}$ such that~$\mathbf{K}=\mathbf{V}\bm{\Lambda}\mathbf{V}^{\He}$ with~$\mathbf{V}$ a unitary matrix of size~$M\times M$ and $\mathbf{\Lambda}=\diag(\lambda_1,\lambda_2,\ldots,\lambda_K)$ such that~$\lambda_1\geq \lambda_2\geq \ldots,\geq \lambda_K$. We then write
\begin{align}
&	\E_{\tilde{\mathbf{H}}_2} \log \det (\Id_{N_2}+\mathbf{H}_2 \mathbf{K} \mathbf{H}_2^\He)\notag\\
&=\E_{\tilde{\mathbf{H}}_2}\log \det (\Id_{N_2}+\mathbf{H}_2\mathbf{K} \mathbf{H}_2^\He)+\E_{\tilde{\mathbf{H}}_2} \log \det (\Id_{N_2}+\mathbf{H}_2\mathbf{H}_2^\He)- \E_{\tilde{\mathbf{H}}_2} \log \det (\Id_{N_2}+\mathbf{H}_2 \mathbf{H}_2^\He)\label{eq:App_Out_1_1}\\
&\geq \E_{\tilde{\mathbf{H}}_2} \log \det (\Id_{N_2}+\mathbf{H}_2\mathbf{H}_2^\He+(\Id_{N_2}+\mathbf{H}_2\mathbf{H}_2^\He)^{\frac{1}{2}}\mathbf{H}_2\mathbf{K}\mathbf{H}_2^\He\left((\Id_{N_2}+\mathbf{H}_2\mathbf{H}_2^\He)^{\frac{1}{2}}\right)^{\He})\notag\\
&~~~~~~~~~~~~~~~~~~~~~~~~~~~~~~~~~~~~~~~~~~~~~~~~~~~~~~~~~~~~~~~- N_2\E_{\tilde{\mathbf{H}}_2} \log \det (1+\|\mathbf{H}_2\|^2_{\mathrm{F}})\label{eq:App_Out_1_2}\\
&\geq\E_{\tilde{\mathbf{H}}_2} \log \det (\Id_{N_2}+\mathbf{H}_2(\Id_M+\mathbf{K})\mathbf{H}_2^\He)- N_2\log \det (1+\|\hat{\mathbf{H}}_2\|^2_{\mathrm{F}}+MN_2\sigma^2)\label{eq:App_Out_1_3}\\
&=\E_{\tilde{\mathbf{H}}_2} \log \det (\Id_{N_2}+\mathbf{H}_2(\Id_M+\mathbf{K})\mathbf{H}_2^\He)+O(1)\label{eq:App_Out_1_4}
\end{align} 
where \eqref{eq:App_Out_1_3} has been obtained by applying Jensen's inequality. We define $\bm{\Lambda}'=\diag(\lambda_{1},\lambda_{2},\ldots,\lambda_{N_1})$ as the matrix containing the $N_1$ largest eigenvalues from~$\mathbf{\Lambda}$ and we proceed from \eqref{eq:App_Out_1_4} as
	\begin{align}  
&\E_{\tilde{\mathbf{H}}_2} \log \det (\Id_{N_2}+\mathbf{H}_2 \mathbf{K} \mathbf{H}_2^\He)\notag\\
&~~~\geq\E_{\tilde{\mathbf{H}}_2} \log \det (\Id_{N_2}+ \mathbf{H}_2 \mathbf{V} (\Id_M+\mathbf{\Lambda}) \mathbf{V}^{\He}\mathbf{H}_2^{\He} )+O(1)\label{eq:App_Out_2_1}\\  
&~~~\stackrel{a}{\ge}\E_{\tilde{\mathbf{H}}_2} \log \det (\Id_{N_2}+ \mathbf{H}_2 \mathbf{V}' (\Id_{N_1}+\mathbf{\Lambda}') \mathbf{V}'^{\He}\mathbf{H}_2^{\He} )+O(1)\label{eq:App_Out_2_2}\\  
&~~~=\E_{\tilde{\mathbf{H}}_2} \log \det (\Id_{N_2}+ \mathbf{\Phi}' (\Id_{N_1}+\mathbf{\Lambda}') \mathbf{\Phi}'^\He )+O(1)\label{eq:App_Out_2_3}\\
&~~~\stackrel{b}{\ge} \frac{N_2}{N_1} \log \det (\Id_{N_1}+\mathbf{\Lambda}') +\frac{N_2(N_1-N_2)}{N_1} \log(\sigma^2)+O(1)\label{eq:App_Out_2_4}
 \end{align}
where we have defined ${\bm{\Phi}}'\triangleq\mathbf{H}_2 \mathbf{V}'\in \mathbb{C}^{{N_2} \times {N_1}}$ with $\mathbf{V}'$ containing the $N_1$ largest eingenvectors, i.e., such that
\begin{equation}
\mathbf{K}=\mathbf{V}'\mathbf{\Lambda}'(\mathbf{V}')^{\He}+(\mathbf{V}-\mathbf{V}')(\mathbf{\Lambda}-\mathbf{\Lambda}')(\mathbf{V}-\mathbf{V}')^{\He}.
\label{eq:App_Out_3}
\end{equation}
 Inequality~$a$ follows from the fact that~$\det(\I+\mathbf{X})\geq\det(\I+\mathbf{Y})$ if $\mathbf{X}\succeq \mathbf{Y}$. Inequality~$b$ is verified because the Gaussian distribution remains invariant by multiplication with a deterministic rotation. Hence, ${\bm{\Phi}}'$ can be written as~$\hat{\bm{\Phi}}'+\tilde{\bm{\Phi}}'$ with the elements of~$\tilde{\bm{\Phi}}$ distributed as the elements of~$\tilde{\mathbf{H}}_2$.

As a consequence, the following lemma presented in~\cite{Yi2012} (although in a different form) can be applied to obtain inequality~$b$.
\begin{lemma}\label{lemma:caseB}
Given a random matrix $\mathbf{H} = \hat{\mathbf{H}}+\tilde{\mathbf{H}} \in \mathbb{C}^{n \times m}$ $(n \leq m \leq 2 n)$, where $\tilde{\mathbf{H}}$ is independent of $\hat{\mathbf{H}}$ and has its entries distributed as i.i.d. $\CN(0,\sigma^2)$, and any $\mathbf{K} \succeq 0$ with eigenvalues $\bm{\Lambda}\triangleq \diag([\lambda_1,\lambda_2,\ldots,\lambda_{m}])$, with $\lambda_1 \ge \lambda_2\ge \cdots \ge \lambda_{m} \ge 0$, it holds that
\begin{align}
\E_{\tilde{\mathbf{H}}} \log \det (\Id_n+\mathbf{H} \mathbf{K} \mathbf{H}^\He) \ge \frac{n}{m} \log \det (\bm{\Lambda}) +\frac{n(m-n)}{m} \log(\sigma^2)  + O(1).
\end{align} 
\end{lemma}

\item We now turn to deriving an upper bound for the first term in \eqref{eq:Lemma_Out}.
\begin{align}
 \frac{1}{N_1}\E_{\tilde{\mathbf{H}}_1} \log \det (\Id_{N_1}+\mathbf{H}_1 \mathbf{K} \mathbf{H}_1^\He)&\leq \frac{1}{N_1}\E_{\tilde{\mathbf{H}}_1} \sum_{i=1}^{N_1}\log (1+\|\mathbf{H}_1\|_{\mathrm{F}}^2 \lambda_i)\label{eq:App_Out_3_1}\\
&\leq \frac{1}{N_1} \sum_{i=1}^{N_1}\log (1+(\|\hat{\mathbf{H}}_1\|^2_{\mathrm{F}}+MN_1\sigma^2)\lambda_i)+O(1)\label{eq:App_Out_3_2}\\
&\leq \frac{1}{N_1} \sum_{i=1}^{N_1}\log (1+\left[\max(\|\hat{\mathbf{H}}_1\|^2_{\mathrm{F}}+MN_1\sigma^2),1)\right]\lambda_i)+O(1).\label{eq:App_Out_3_3}
\end{align} 
\end{itemize}
From the upper bound~\eqref{eq:App_Out_3_3} and the lower bound \eqref{eq:App_Out_2_4}, we can then write
\begin{align}
 &\frac{1}{N_1}\E_{\tilde{\mathbf{H}}_1} \log \det (\Id_{N_1}+\mathbf{H}_1 \mathbf{K} \mathbf{H}_1^\He) - \frac{1}{N_2} \E_{\tilde{\mathbf{H}}_2} \log \det (\Id_{N_2}+\mathbf{H}_2 \mathbf{K} \mathbf{H}_2^\He) \notag\\
&~~~~\leq \frac{1}{N_1} \sum_{i=1}^{N_1}\left(\log (1\!+\!\left[\max(\|\hat{\mathbf{H}}_1\|^2_{\mathrm{F}}\!+\!MN_1\sigma^2),1)\right]\lambda_i)\!-\!\log(1\!+\!\lambda_i)\right) -\frac{N_1\!-\!N_2}{N_1} \log(\sigma^2)\!+\!O(1)\label{eq:App_Out_4_1}\\
&~~~~= -\frac{N_1-N_2}{N_1} \log(\sigma^2)+O(1)\label{eq:App_Out_4_2}
\end{align}
where \eqref{eq:App_Out_4_2} is obtained by observing that the sum of difference of logarithms in \eqref{eq:App_Out_4_1} remains bounded for any values taken by the $\lambda_i$.
\bibliographystyle{IEEEtran}
\bibliography{Literatur}

\begin{thebibliography}{10}
\providecommand{\url}[1]{#1}
\csname url@samestyle\endcsname
\providecommand{\newblock}{\relax}
\providecommand{\bibinfo}[2]{#2}
\providecommand{\BIBentrySTDinterwordspacing}{\spaceskip=0pt\relax}
\providecommand{\BIBentryALTinterwordstretchfactor}{4}
\providecommand{\BIBentryALTinterwordspacing}{\spaceskip=\fontdimen2\font plus
\BIBentryALTinterwordstretchfactor\fontdimen3\font minus
  \fontdimen4\font\relax}
\providecommand{\BIBforeignlanguage}[2]{{%
\expandafter\ifx\csname l@#1\endcsname\relax
\typeout{** WARNING: IEEEtran.bst: No hyphenation pattern has been}%
\typeout{** loaded for the language `#1'. Using the pattern for}%
\typeout{** the default language instead.}%
\else
\language=\csname l@#1\endcsname
\fi
#2}}
\providecommand{\BIBdecl}{\relax}
\BIBdecl

\bibitem{Telatar1999}
I.~E. Telatar, ``{Capacity of multi-antenna Gaussian channels},''
  \emph{{European Transaction on Communications}}, vol.~10, pp. 585--595, 1999.

\bibitem{Jafar2005}
S.~A. Jafar and A.~J. Goldsmith, ``{Isotropic fading vector broadcast Channels:
  The scalar upper bound and loss in degrees of freedom},'' \emph{IEEE Trans.
  Inf. Theory}, vol.~51, no.~3, pp. 848--857, Mar. 2005.

\bibitem{Jindal2006}
N.~Jindal, ``{MIMO broadcast channels with finite-rate feedback},'' \emph{IEEE
  Trans. Inf. Theory}, vol.~52, no.~11, pp. 5045--5060, Nov. 2006.

\bibitem{Caire2010}
G.~Caire, N.~Jindal, M.~Kobayashi, and N.~Ravindran, ``{Multiuser MIMO
  achievable rates with downlink training and channel state feedback},''
  \emph{{IEEE} Trans. Inf. Theory}, vol.~56, no.~6, pp. 2845--2866, Jun. 2010.

\bibitem{Huang2012}
C.~Huang, S.~Jafar, S.~Shamai, and S.~Vishwanath, ``{On degrees of freedom
  region of MIMO networks without channel state information at transmitters},''
  \emph{{IEEE} Trans. Inf. Theory}, vol.~58, no.~2, pp. 849--857, Feb. 2012.

\bibitem{Vaze2012b}
C.~S. Vaze and M.~K. Varanasi, ``{The degree-of-freedom regions of MIMO
  broadcast, interference, and cognitive radio channels with no CSIT},''
  \emph{{IEEE} Trans. Inf. Theory}, vol.~58, no.~8, pp. 5354--5374, Aug. 2012.

\bibitem{Love2008}
D.~J. Love, R.~W. Heath, V.~K.~N. Lau, D.~Gesbert, B.~D. Rao, and M.~Andrews,
  ``{An overview of limited feedback in wireless communication systems},''
  \emph{IEEE J. Sel. Areas Commun.}, vol.~26, no.~8, pp. 1341--1365, Oct. 2008.

\bibitem{MaddahAli2010}
M.~A. Maddah-Ali and D.~N.~C. Tse, ``{Completely stale transmitter channel
  state information is still very useful},'' in \emph{Proc. {A}llerton
  {C}onference on {C}ommunication, {C}ontrol, and {C}omputing (Allerton)},
  2010.

\bibitem{MaddahAli2012}
M.~Maddah-Ali and D.~Tse, ``{Completely stale transmitter channel state
  information is still very useful},'' \emph{{IEEE} Trans. Inf. Theory},
  vol.~58, no.~7, pp. 4418--4431, Jul. 2012.

\bibitem{Vaze2011a}
C.~S. Vaze and M.~K. Varanasi, ``{The degrees of freedom region of the two-user
  MIMO broadcast channel with delayed CSIT},'' in \emph{Proc. IEEE
  International Symposium on Information Theory (ISIT)}, 2011.

\bibitem{Abdoli2011a}
M.~J. Abdoli, A.~Ghasemi, and A.~K. Khandani, ``{On the degrees of freedom of
  three-user MIMO broadcast channel with delayed CSIT},'' in \emph{Proc. IEEE
  International Symposium on Information Theory (ISIT)}, 2011.

\bibitem{Vaze2012a}
C.~S. Vaze and M.~K. Varanasi, ``{The degrees of freedom region and
  interference alignment for the MIMO interference channel with delayed
  CSIT},'' \emph{{IEEE} Trans. Inf. Theory}, vol.~58, no.~7, pp. 4396--4417,
  Jul. 2012.

\bibitem{Maleki2012}
H.~Maleki, S.~A. Jafar, and S.~Shamai~(Shitz), ``{Retrospective interference
  alignment over interference networks},'' \emph{{IEEE} Journal of Sel. Topics
  in Sign. Process.}, vol.~6, no.~3, Jun. 2012.

\bibitem{Abdoli2011b}
\BIBentryALTinterwordspacing
M.~J. Abdoli, A.~Ghasemi, and A.~K. Khandani, ``{On the degrees of freedom of
  K-user SISO interference and X channels with delayed CSIT},'' 2011, submitted
  to {IEEE} Trans. Inf. Theory. [Online]. Available:
  \url{http://arxiv.org/abs/1109.4314}
\BIBentrySTDinterwordspacing

\bibitem{Tandon2012a}
R.~Tandon, S.~Mohajer, V.~Poor, and S.~Shamai, ``{Degrees of freedom region of
  the MIMO interference channel with output feedback and delayed CSIT},''
  \emph{{IEEE} Trans. Inf. Theory}, vol.~PP, no.~99, p.~1, 2012.

\bibitem{Mohanty2012}
\BIBentryALTinterwordspacing
K.~Mohanty, C.~S. Vaze, and M.~K. Varanasi, ``{The degrees of freedom region
  for the MIMO interference channel with hybrid CSIT},'' 2012, submitted to
  IEEE Trans. Wireless Commun. [Online]. Available:
  \url{http://arxiv.org/abs/1209.0047}
\BIBentrySTDinterwordspacing

\bibitem{Lee2012}
N.~Lee and R.~W. Heath, ``{Not too delayed CSIT achieves the optimal degrees of
  freedom},'' in \emph{Proc. {A}llerton {C}onference on {C}ommunication,
  {C}ontrol, and {C}omputing (Allerton)}, 2012.

\bibitem{Tandon2012b}
\BIBentryALTinterwordspacing
R.~Tandon, S.~A. Jafar, S.~Shamai, and H.~V. Poor, ``{On the synergistic
  benefits of alternating CSIT for the MISO BC},'' 2012, submitted to IEEE
  Trans. Inf. Theory. [Online]. Available: \url{http://arxiv.org/abs/1208.5071}
\BIBentrySTDinterwordspacing

\bibitem{Kobayashi2012}
M.~Kobayashi, S.~Yang, D.~Gesbert, and X.~Yi, ``{On the degrees of freedom of
  time correlated MISO broadcast channel with delayed CSIT},'' in \emph{{Proc.
  IEEE International Symposium on Information Theory (ISIT)}}, 2012.

\bibitem{Yang2012}
S.~Yang, M.~Kobayashi, D.~Gesbert, and X.~Yi, ``{Degrees of freedom of time
  correlated MISO broadcast channel with delayed CSIT},'' \emph{{IEEE} Trans.
  Inf. Theory}, vol.~PP, no.~99, p.~1, 2012.

\bibitem{Gou2012}
T.~Gou and S.~Jafar, ``{Optimal use of current and outdated channel state
  information: Degrees of Freedom of the MISO BC with mixed CSIT},''
  \emph{{IEEE Communications Letters}}, vol.~16, no.~7, pp. 1084--1087, Jul.
  2012.

\bibitem{Chen2012a}
J.~{C}hen and P.~{E}lia, ``{C}an imperfect delayed {CSIT} be as useful as
  perfect delayed {CSIT}? {D}o{F} analysis and constructions for the {BC},'' in
  \emph{Proc. {A}llerton {C}onference on {C}ommunication, {C}ontrol, and
  {C}omputing (Allerton)}, 2012.

\bibitem{Chen2012b}
\BIBentryALTinterwordspacing
J.~Chen and P.~Elia, ``{MISO broadcast channel with delayed and evolving
  CSIT},'' 2012, submitted to IEEE Trans. Inf. Theory. [Online]. Available:
  \url{http://arxiv.org/abs/1211.1622}
\BIBentrySTDinterwordspacing

\bibitem{Yi2012}
\BIBentryALTinterwordspacing
X.~Yi, S.~Yang, D.~Gesbert, and M.~Kobayashi, ``{The degrees of freedom region
  of temporally-correlated MIMO networks with delayed CSIT},'' 2012, submitted
  to IEEE Trans. Inf. Theory. [Online]. Available:
  \url{http://arxiv.org/abs/1211.3322}
\BIBentrySTDinterwordspacing

\bibitem{Xu2012}
J.~Xu, J.~G. Andrews, and S.~A. Jafar, ``{MISO broadcast channels with delayed
  finite-rate feedback: Predict or observe?}'' \emph{IEEE Trans. on Wireless
  Commun.}, vol.~11, no.~4, pp. 1456--1467, Apr. 2012.

\bibitem{cover2006}
T.~Cover and A.~Thomas, \emph{{Elements of information theory}}.\hskip 1em plus
  0.5em minus 0.4em\relax Wiley-Interscience, Jul. 2006.

\bibitem{Liu2007}
T.~Liu and P.~Viswanath, ``An extremal inequality motivated by multiterminal
  information-theoretic problems,'' \emph{IEEE Trans. Inf. Theo.}, vol.~53,
  no.~5, pp. 1839--1851, May 2007.

\bibitem{Weingarten2006}
H.~Weingarten, Y.~Steinberg, and S.~Shamai, ``The capacity region of the
  gaussian multiple-input multiple-output broadcast channel,'' \emph{IEEE
  Trans. Inf. Theo.}, vol.~52, no.~9, pp. 3936--3964, Sep. 2006.

\end{thebibliography}
\end{document}